\documentclass[letterpaper, 11pt]{article}

\usepackage[margin=1in]{geometry}
\usepackage{amsmath}
\usepackage{amssymb}
\usepackage{amsthm}
\usepackage[colorlinks,urlcolor=blue,citecolor=blue,linkcolor=blue]{hyperref}
\usepackage{algorithm}
\usepackage{algorithmic}
\usepackage{multicol}
\usepackage{paralist}
\usepackage[usenames, dvipsnames]{color}
\usepackage{xfrac}

\newtheoremstyle{mythm}
     {3pt}
     {3pt}
     {\itshape}
     {}
     {\bfseries}
     {.}
     { }
     {}
 \theoremstyle{mythm}

\newtheorem{thm}{Theorem}[section]
\newtheorem{lem}[thm]{Lemma}
\newtheorem{cor}[thm]{Corollary}

\newtheorem{dfn}[thm]{Definition}

\newtheorem{que}[thm]{Question}

\newcommand\numberthis{\addtocounter{equation}{1}\tag{\theequation}}

\newcommand{\dft}[1]{\textbf{\textit{#1}}}

\newcommand{\calD}{\mathcal{D}}

\DeclareMathOperator*{\E}{\mathbf{E}}
\newcommand{\R}{\mathbf{R}}

\newcommand{\e}{\varepsilon}
\newcommand{\eps}{\varepsilon}
\newcommand{\sO}{\tilde{O}}

\newcommand{\Dd}{\calD_d}
\newcommand{\NW}{\mathsf{NW}}

\newcommand{\abs}[1]{\left|#1\right|}
\newcommand{\norm}[1]{\left\|#1\right\|}
\newcommand{\paren}[1]{\left(#1\right)}
\newcommand{\set}[1]{\left\{#1\right\}}

\renewcommand{\th}{^{\textrm{th}}}

\newcommand{\thres}{\sqrt{m/\eps}}

\newcommand{\thr}{\theta}

\newcommand{\sle}{\hyperref[sle]{\textup{\color{black}{\sf Sample-light-edge}}}}
\newcommand{\she}{\hyperref[she]{\textup{\color{black}{\sf Sample-heavy-edge}}}}
\newcommand{\saue}{\hyperref[saue]{\textup{\color{black}{\sf Sample-edge-almost-uniformly}}}}

\DeclareMathOperator{\tvd}{dist_{TV}}

\newcommand{\mA}{\mathcal{A}}

\newcommand{\mhat}{\widehat{m}}

\newcommand{\mL}{\mathcal{L}}
\newcommand{\emh}{E_{\mH}}
\newcommand{\eml}{E_{\mL}}

\newcommand{\mH}{\mathcal{H}}

\newcommand{\dm}{d^{\mL}}
\newcommand{\dpl}{d^{\mH}}

\newcommand{\setq}{\frac{10n}{(1-\eps)\sqrt{\eps \mhat}}}

\def\withcolors{1}
\ifnum\withcolors=1

\else

\fi

\title{On Sampling Edges Almost Uniformly}

\author{%
  \begin{tabular}{c@{\extracolsep{16pt}}c}
    Talya Eden\thanks{\texttt{talyaa01@gmail.com}} & Will Rosenbaum\thanks{
      \texttt{will.rosenbaum@gmail.com}
    }
    \\
    \multicolumn{2}{c}{School of Electrical Engineering}\\
    \multicolumn{2}{c}{Tel Aviv University}\\ 
    \multicolumn{2}{c}{Tel Aviv 6997801}\\ 
    \multicolumn{2}{c}{Israel}
  \end{tabular}
}

\date{\today}

\begin{document}

  \maketitle

  \begin{abstract}
    We consider the problem of sampling an edge almost uniformly from an unknown graph, $G = (V, E)$. Access to the graph is provided via queries of the following types: (1) uniform vertex queries, (2) degree queries, and (3) neighbor queries. We describe an algorithm that returns a random edge $e \in E$ using $\tilde{O}(n / \sqrt{\e m})$ queries in expectation, where $n = \abs{V}$ is the number of vertices, and $m = \abs{E}$ is the number of edges, such that each edge $e$ is sampled with probability $(1 \pm \e)/m$. We prove that our algorithm is optimal in the sense that any algorithm that samples an edge from an almost-uniform distribution must perform $\Omega(n / \sqrt{m})$ queries.
  \end{abstract}

  \newcommand{\strikeout}[1]{\textcolor{Gray}{#1}}
\newcommand{\tcom}[1]{\textcolor{red}{ #1}}
 \section{Introduction}

  The omnipresence of massive data sets motivates the study of \emph{sublinear algorithms}---algorithms for which the available resources (time, space, etc.) are extremely limited relative to the size of the algorithm's input. We focus on the case where time is the limited resource. In sublinear time, an algorithm cannot access its entire input before producing output. Thus, standard models of computation for sublinear time algorithms provide query access to the input. Two natural questions are, \emph{``What queries should the model allow?''} and, \emph{``How many queries are necessary to solve a particular problem?''} While the answer to the first question may be dictated in practice by various circumstances, it is interesting to theoretically understand the relative computational power afforded by different allowable queries.

  We consider algorithms whose inputs are simple undirected graphs, $G = (V, E)$. We allow an algorithm to access $G$ via queries of the following types: (1)~sample a uniformly random vertex (\emph{vertex queries}), (2)~query the degree of a given vertex\footnote{We note that degree queries can be implemented by performing $O(\log n)$ neighbor queries per degree query.} (\emph{degree queries}), and (3)~query for the $i\th$ neighbor of a vertex $v$ (\emph{neighbor queries}). Our lower bound applies to a strictly more powerful model that additionally allows an algorithm to ask if two vertices $u, v$ form an edge (\emph{pair queries}). A primary goal of sublinear time graph algorithms is to determine the number of queries necessary and sufficient to solve a given graph problem. Since the number of queries may be a random variable, we examine the \emph{expected} number of queries needed to solve a problem.

  Using the model above, we ask, \emph{``How many queries are required to sample an edge from $E$ uniformly at random?''} We show that in general, $\tilde{\Theta}(n / \sqrt{m})$ queries are both necessary and sufficient to sample an edge almost uniformly, independent of the graph's topology. We remark that the problem of sampling an edge uniformly at random is equivalent to sampling a random vertex where the probability of sampling $v \in V$ is proportional to its degree, $d(v)$.

  An algorithm for sampling an edge in a graph almost uniformly was first suggested  by~\cite{Kaufman2004}. In \cite{Kaufman2004}, Kaufman et al.\ use random edge samples in order to test if a graph is bipartite. In particular, they devise a subroutine that guarantees that all but a small fraction of the edges are each sampled with probability $\Omega(1/m)$. However, their subroutine never samples edges such that both of their endpoints are high-degree vertices, i.e., have degree above a certain threshold $\theta$. While the fraction of such edges is small for the correct setting of $\theta$, these edges might contain invaluable information about the graph. Therefore, it is desirable to sample edges from a distribution that is close to uniform in a stronger sense.

  \subsection{Results}
  
  We describe a sublinear time  algorithm for sampling an edge in a graph such that each edge is sampled with almost equal probability. We further prove that that the query complexity of our algorithm is essentially optimal. Our main result shows that it is possible to sample edges from a distribution which is \emph{pointwise $\e$-close to uniform}, a notion we formalize in the following definition.

  \begin{dfn}
    \label{dfn:epsilon-close}
    Let $Q$ be a fixed probability distribution on a finite set $\Omega$. We say that a probability distribution $P$ is \dft{pointwise  $\e$-close to $Q$} if for all $x \in \Omega$,
    \[
    \abs{P(x) - Q(x)} \leq \e Q(x)\,,\quad\text{or equivalently}\quad 1 - \e \leq \frac{P(x)}{Q(x)} \leq 1 + \e\,.
    \]
    If $Q = U$, the uniform distribution on $\Omega$, then we say that $P$ is \dft{pointwise  $\e$-close to uniform}.
  \end{dfn}

  The above definition implies that if a distribution $P$ is pointwise $\eps$-close to a distribution $Q$ then $P$ is close to $Q$ under many important distance measures. In particular, it holds that the support of $P$ and $Q$ are equal and that the total variational distance between $P$ and $Q$ is at most $\eps$.

%

  	
  \begin{thm}
    \label{thm:alg}
    Let $G = (V, E)$ be an arbitrary graph with $n$ vertices and $m$ edges. There exists an algorithm which uses $\sO(n / \sqrt{\eps m})$ vertex, degree, and neighbor queries in expectation, and with probability at least $2/3$, returns an edge. Each returned edge is sampled according to a distribution $P$ that is pointwise $\eps$-close to uniform. The runtime of the algorithm is linear in the number of queries.
  \end{thm}

  We show that the algorithm of Theorem~\ref{thm:alg} is optimal in the sense that any algorithm which returns an edge from $E$ almost uniformly must use $\Omega(n / \sqrt{ m})$ queries. This lower bound holds even if we allow pair queries and only requires the sample distribution to be close to uniform in total variational distance (Definition~\ref{def:tvd}).\footnote{Closeness in total variational distance is strictly weaker than pointwise closeness.} 
	
  \begin{thm}
    \label{thm:lb}
    Let $\mA$ be an algorithm that performs $q$ vertex, degree, neighbor, or pair queries and with probability at least $2/3$ returns an edge $e \in E$ sampled according to a distribution $P$ over $E$. If for all $G =(V, E)$, $P$ satisfies $\tvd(P, U) < \e$, then $q = \Omega(n / \sqrt{m})$. This lower bound holds if $\mA$ knows precisely $m$ and $n$.
  \end{thm}

  
  The strength of the approximation guarantee of the algorithm in Theorem~\ref{thm:alg} allows us to obtain the following corollary for sampling weighted edges in graphs. 
  
  \begin{cor}
    \label{cor:weighted-edges}
    Let $G = (V, E)$ and  $\e$ be as in Theorem~\ref{thm:alg}, and let $w : E \to \R$ be an arbitrary function. Let $P$ be the distributions on edges induced by the algorithm of Theorem~\ref{thm:alg}. Then
    \[
    \abs{\E_{e \sim P}(w(e)) - \E_{e \sim U}(w(e))} \leq \e \abs{\E_{e \sim U}(w(e))}.
    \]
  \end{cor}

  The weaker requirement of $\tvd(P,U) < \e$ would not have been sufficient to obtain the above corollary.\footnote{Consider, for example, a distribution $P$ such that $\tvd(P,U) < \e$  and a weight function $w$ that is non zero on a single edge $e$ which is missing from the support of $P$.} 
We also obtain the following corollary.

  \begin{cor} 
    \label{cor:deg-dist}
    Let $G = (V, E)$ be an arbitrary graph with $n$ vertices and $m$ edges, and $\e > 0$. Let $\Dd$ denote the degree distribution of $G$ (i.e., $\Dd(v) = d(v) / 2m$). Then there exists an algorithm that with probability at least $2/3$ produces a vertex $v$ sampled from a distribution $P$ over $V$ which is pointwise $\e$-close to $\Dd$ using $\sO(n  / \sqrt{\eps m})$ vertex, degree and neighbor queries.
  \end{cor}

  \subsection{Overview of the algorithm}

  Consider the process of sampling a vertex uniformly and selecting one of its neighbors uniformly at random. Each (directed) edge $(u,v)$ is sampled with probability $1/(n\,d(u))$, thus biasing the sample towards edges originating from vertices with low degrees. In order to overcome this bias, we partition the edges into two sets according to the degrees of their endpoints, and sample edges independently from each set. Combining the two procedures we guarantee that all edges are sampled with almost equal probability.
  
  Let $d(v)$ denote the degree of the vertex $v$. For a fixed threshold $\thr$, we call a vertex $v$ \emph{light} if $d(v) \leq \thr$, and \emph{heavy} otherwise. We treat each undirected edge as a pair of directed edges, and call a directed edge $(u, v)$ light (resp.\ heavy) if $u$ is light (resp.\ heavy). We would like to set $\thr = \thres$, although if $m$ is not known to the algorithm, then a constant factor approximation of of $m$ suffices.

  In order to sample a light edge we first sample a uniform vertex, denoted $u$. If $u$ is heavy, we fail. Otherwise  $d(u)\leq \thr$, and we choose an index $i$ in $[\thr]$ uniformly at random, and query for the $i\th$ neighbor of $u$. Since $\theta$ in only an upper bound on $d(v)$, this query could fail. Nonetheless, every light edge $(u,v)$ is sampled with probability $1 / (n \,\thr)$, regardless of $u$'s degree.

  In order to sample a heavy edge we use the following procedure. We first sample a light edge $(u, v)$ as described above. If $v$ is heavy we then query for a random neighbor $w$ of $v$. The probability of hitting some specific heavy edge $(v,w)$ is $\frac{d^{\mL}(v)} {n\,\thr}\cdot \frac{1}{d(v)}$, where $d^{\mL}(v)$ denotes the number of light neighbors of $v$. The threshold $\theta$ is set as to ensure that for every heavy vertex, at most an $\eps$-fraction of its neighbors are heavy. It follows that  $d^{\mL}(v)\approx d(v)$, and thus each heavy edge $(v,w)$ is chosen with probability roughly $1 / (n\,\thr)$.

  Our algorithm invokes each of the procedures above with equal probability sufficiently many times. We argue that with high constant probability an edge is returned, and that the induced distribution on edges is pointwise $\eps$-close to uniform.

  \subsection{Related work}

  \textbf{Sampling edges.} As discussed previously, the first reference that specifically poses the question of sampling edges uniformly  in sublinear time is~\cite{Kaufman2004}. More recently, in~\cite{Eden2015}, Eden et al.\ use random edge sampling as a subroutine in their algorithm for approximating the number of triangles in a graph. However, they avoid the ``difficult task of selecting random edges from the entire graph,'' by instead sampling from a smaller subgraph. A similar workaround to random edge sampling is employed in~\cite{Eden2016}, where the authors use subgraph edge sampling to approximate moments of the degree distribution of a graph. Recently, Aliakbarpour et al.~\cite{Aliakbarpour2017} used random edge sampling to estimate the number of stars in a graph. The authors show that allowing uniform random edge queries (in addition to vertex queries) affords strictly greater computational power, as their algorithm outperforms the lower bound of Gonen et al.~\cite{Gonen2011} whose model only allows vertex, degree, and neighbor queries. Despite the appearance of edge sampling in these works, we are unaware of any previous work that explicitly addresses the query complexity of sampling uniformly random edges.

  \textbf{Approximating the number of edges.} The problem of uniformly sampling an edge in a graph is closely related to approximating the number of edges in the graph. The first sublinear algorithm for approximating the number of edges in a graph is due to Feige~\cite{Feige2006}. Feige's algorithm uses only vertex and  degree queries to obtain a $(2+\e)$-factor approximation of $m$ using $\sO(n/\sqrt m)$ queries. Further, Feige proves that this algorithm is optimal, as any algorithm which computes a $(2 - o(1))$-factor approximation \emph{using only vertex and degree queries} requires $\Omega(n)$ queries. Goldreich \& Ron~\cite{Goldreich2008} described a $(1 + \e)$-factor approximation of $m$ using $\sO(n/\sqrt{m})$ queries by allowing neighbor queries in addition to vertex and degree queries. Both of these algorithms have a polynomial dependence in $1/\eps$.

  \textbf{Sampling from a different distribution.}  The problem of sampling a uniformly random edge in $G$ is equivalent to sampling a random vertex where each vertex $v$ is chosen with probability proportional to its degree, denoted $d(v)$.  To see this, observe that each $v \in V$ appears in $d(v)$ edges. Thus, sampling a random edge, then returning each of its incident vertices with probability $1/2$ will return a vertex $v \in V$ with probability $d(v) / 2 m$. Sampling from non-uniform distributions is a central problem in statistics, and has previously been employed in the study of sublinear algorithms. In particular, sampling an object $i$ with probability proportional to its weight $w_i$ is known as \emph{linear weighted sampling} (LWS)~\cite{Motwani2007} and \emph{importance sampling}~\cite{Liu1996}. Motwani et al.~\cite{Motwani2007} used LWS to estimate $\sum_i w_i$ with a sublinear number of samples. Batu et al.~\cite{Batu2009} used LWS in a sublinear approximate bin packing scheme. For general weights, it is easy to see that a linear number of uniform samples is required to obtain a single linear weighted sample.\footnote{Consider, for example, weights $w_i$ where $w_j = 1$ and $w_i = 0$ for all $i \neq j$.} However, our algorithm provides a linear weighted sample using a sublinear number of uniform samples when the weights are the degrees of the vertices in a given graph.



  \textbf{Sampling from real-world networks.} Sampling from large graphs is a matter of practical importance in the study of real-world networks (e.g., social, economic, and computer networks). In these contexts, random sampling is often used to estimate parameters of the network, or to produce a smaller network with similar statistical features to the original network~\cite{Leskovec2006, Ahmed2013}. Prior sampling strategies include sampling random vertices~\cite{Leskovec2006, Hubler2008} or edges~\cite{Ahmed2011, Leskovec2006, Ribeiro2010}, strategies based on random walks~\cite{Leskovec2006, Ribeiro2010}, or a mixture of these strategies~\cite{Jin2011}. Since different sampling techniques generally have different associated costs and values depending on the particular network being tested and the desired output~\cite{Ribeiro2010}, it is useful to understand the relative query complexity of different sets of queries. Studies such as~\cite{Leskovec2006} and~\cite{Ribeiro2010} are motivated by practical considerations in deciding the types of queries to be used. Our results complement these works by giving theoretical guarantees for the complexity of simulating edge samples from vertex samples.

    \section{Preliminaries}

  \paragraph{Graphs.} Let $G=(V,E)$ be an undirected graph with $n = \abs{V}$ vertices. We treat each undirected edge $\set{u, v}$ in $E$ as pair of directed edges $(u, v)$, $(v, u)$ from $u$ to $v$ and from $v$ to $u$, repsectively. We denote the (undirected) degree of a vertex $v \in V$ by $d(v)$. Thus, the number of (directed) edges in $G$ is $m = \sum_{v \in V}d(v)$. For $v \in V$, we denote the set of neighbors of $v$ by $\Gamma(v)$. We assume that each $v$ has an arbitrary but fixed order on $\Gamma(v)$ so that we may refer unambiguously to $v$'s $i\th$ neighbor for $i = 1, 2, \ldots, d(v)$.

  We partition $V$ and $E$ into sets of light and heavy elements depending on their degrees.

  \begin{dfn}
    \label{def:heavy-light}
    For a given threshold $\thr$, we say that a vertex $u$ is a \dft{light vertex} if $d(u) \leq \thr$ and otherwise we say that it is a \dft{heavy vertex}. We say that an edge is a \dft{light edge} if it originates from a light vertex. A \dft{heavy edge} is defined analogously. Finally, we denote the sets of light and heavy vertices by $\mL$ and $\mH$ respectively, and the sets of light and heavy edges by $\eml$ and $\emh$ respectively.
  \end{dfn}

  \paragraph{Queries and Complexity.} The algorithms we consider access a graph $G$ via queries. Our algorithm uses the following queries:
  \begin{enumerate} 
  \item[1.] {Vertex query:} returns a uniformly random vertex $v \in V$.
  \item[2.] {Degree query:} given a vertex $v \in V$, returns $d(v)$.
  \item[3.] {Neighbor query:} given $v \in V$ and $i \in \mathbf{N}$, return $v$'s $i\th$ neighbor; if $i > d(v)$, this operation returns \emph{fail}.
  \end{enumerate}
  Our lower bounds apply additionally to a computational model that allows pair queries.
  \begin{enumerate}
  \item[4.] {pair query:} given $v, w \in V$, returns \emph{true} if $(v, w) \in E$; otherwise returns \emph{false}.
  \end{enumerate}
  
  The query cost of an algorithm $\mA$ is the (expected) number of queries that $\mA$ makes before terminating. We make no restrictions on the computational power of $\mA$ except for the number of queries $\mA$ makes to $G$. The query complexity of a task is the minimum query cost of any algorithm which performs the task.

\paragraph{Distance Measures.} We first recall the definition of the total variational distance between two distributions.
  \begin{dfn} \label{def:tvd}
  	Let $P$ and $Q$ be probability distributions over a finite set $\Omega$. We denote the \dft{total variational distance} or \dft{statistical distance} between $P$ and $Q$ by
  \[
  \tvd(P, Q) = \frac 1 2 \norm{P - Q}_1 = \frac 1 2 \sum_{x \in \Omega} \abs{P(x) - Q(x)}.
  \]
\end{dfn}

  \begin{lem}
    \label{lem:epsilon-close}
    Let $P$ be a probability distribution over a finite set $\Omega$ which satisfies
    \[
      1 - \e \leq \frac{P(x)}{P(y)} \leq 1 + \e \quad\text{for all } x, y \in \Omega\,.
    \]
    Then $P$ is pointwise  $\e$-close to uniform (recall Definition~\ref{dfn:epsilon-close}).
  \end{lem}
  \begin{proof}
    Suppose $P$ satisfies the hypothesis of the lemma. Then
    \[
      1 - \e \leq \frac{P(x)}{P(y)} \leq 1 + \e \implies (1 - \e) \frac{P(y)}{U(x)} \leq \frac{P(x)}{U(x)} \leq (1 + \e) \frac{P(y)}{U(x)}.
    \]
    Summing the second expression over all $y \in \Omega$ gives
    \[
    (1 - \e) \frac{1}{U(x)} \leq \frac{P(x)}{U(x)} \cdot \frac{1}{U(x)} \leq (1 + \e) \frac{1}{U(x)}\,.
    \]
    The factor of $1/U(x)$ appears in the middle term because we sum over $\abs{\Omega} = 1 / U(x)$ terms. Hence $P$ is $\e$-close to uniform.
  \end{proof}

\paragraph{Approximating the number of edges.} In order to choose an appropriate value of the threshold $\thr$ (recall Definition~\ref{def:heavy-light}), we require a suitable estimate of $m$. Such an estimate can be obtained by applying an algorithm due to Goldreich \& Ron~\cite{Goldreich2008}.

\begin{thm}[Goldreich \& Ron~\cite{Goldreich2008}]
  \label{thm:gr08} 
  Let $G = (V, E)$ be a graph with $n$ vertices and $m$ edges. There exists an algorithm which uses $\sO(n / \sqrt{m})$ vertex, degree, and neighbor queries in expectation and outputs an estimate $\mhat$ of $m$ that with probability at least $2/3$ satisfies $m \leq \mhat \leq 2 m$. 
\end{thm}

  \section{Algorithms}
  \label{sec:alg}

  In this section, we present our main algorithm that samples a random edge in $E$ almost uniformly. The main routine of our algorithm---$\saue$---takes the number of edges $m$ as a parameter. In the proof of Theorem~\ref{thm:alg}, we argue that a constant factor approximation of $m$ suffices and invoke the algorithm of Goldreich \& Ron (Theorem~\ref{thm:gr08}) to produce such an approximation. The main algorithm consists of two subroutines. The first, $\sle$, returns a uniform light edge (or fails), while $\she$ returns an almost uniform heavy edge (or fails).
  
  
  \begin{figure}[htb!] \label{sle}
    \fbox{
      \begin{minipage}{0.9\textwidth}
	$\sle(\thr)$
	\smallskip
	\begin{compactenum}
	\item Sample a vertex $u \in V$ uniformly at random and query for its degree. \label{step:l1}
	\item If $d(u) > \thr$ \textbf{return} \emph{fail}.
	\item Choose a number $j \in \left[\thr\right]$ uniformly at random.
	\item Query for the $j\th$ neighbor of $u$.\label{step:l_nbr}
	\item If no vertex was returned then \textbf{return} \emph{fail}. Otherwise, let $v$ be the returned vertex.
	\item \textbf{Return} $(u,v)$.
	\end{compactenum}
      \end{minipage}
    }
  \end{figure}

  \begin{lem}
    \label{lem:light}
    The procedure $\sle$ performs a constant number of queries and succeeds with probability $\abs{\eml}/n \thr$. In the case where $\sle$ succeeds, the edge returned is uniformly distributed in $\eml$.
  \end{lem}

  
  \begin{proof} Suppose a light vertex $u$ is sampled in Step~\ref{step:l1} of the procedure. Then the probability that we obtain a neighbor of $u$ in Step~\ref{step:l_nbr} is $d(u)/\thr$. Hence,
    \[
    \Pr[\text{ Success }] = \sum_{u \in \mL}\frac{1}{n}\cdot \frac{d(u)}{\thr}=\frac{|\eml|}{n\thr}\,.
    \]
    In any invocation of the algorithm, the probability that a particular (directed) edge $e$ is returned is $1 / (n\, \thr)$ if $e$ is light, and $0$ otherwise. Thus, each light edge is returned with equal probability.
  \end{proof}

  \begin{figure}[htb!] \label{she}
    \fbox{
      \begin{minipage}{0.9\textwidth}
	$\she(\thr)$
	\smallskip 
	\begin{compactenum}
	\item Sample a vertex $u \in V$ uniformly at random and query for its degree.
	\item If $d(u) > \thr$ \textbf{return} \emph{fail}.
	\item Choose a number $j \in \left[\thr\right]$ uniformly at random.
	\item Query for the $j\th$ neighbor of $u$.
	\item If no vertex was returned or if the returned vertex is light then \textbf{return} \emph{fail}. Otherwise, let $v$ be the returned vertex. \label{step:h_v}
	\item Sample $w$ a random neighbor of $v$.  \label{step:h_w}
	\item\textbf{Return} $(v,w)$.
	\end{compactenum}
      \end{minipage}
    }
  \end{figure}

  \begin{lem}
    \label{lem:heavy}
    	The procedure $\she$ performs a constant number of queries and succeeds with probability  in $\left[\frac{\abs{\emh}}{n \, \thr} \paren{1 - \frac{m}{\thr^2}}, \frac{\abs{\emh}}{n \, \thr} \right]$. In the case where $\she$ succeeds, the edge returned is distributed according to a distribution $P$ that is pointwise $(m / \thr^2)$-close to uniform.
    
  \end{lem}
  \begin{proof}
    Since each $v \in \mH$ satisfies $d(v) > \thr$, we have $\abs{\mH} < m / \thr$. For every vertex $v\in \mH$ let $d^{\mL}(v)$ denote the cardinality of $\Gamma(v)\cap \mL$, the set of light neighbors of $v$. Similarly, let $d^{\mH}(v)$ denote $|\Gamma(v)\cap \mH|$. Thus
    \[
    \dpl(v) \leq |\mH| < \frac{m}{\thr} < \frac{m}{\thr} \cdot \frac{d(v)}{\thr}.
    \]
    Since  $\dm(v) + \dpl(v) = d(v)$ for every $v$, we have
    \begin{equation}
      \label{eqn:light-degree}
      \dm(v) > \paren{1-\frac{m}{\thr^2}}d(v)\,.
    \end{equation}
    Each vertex $v \in \mH$ is chosen in Step~\ref{step:h_v} with probability $\dm(v)/(n \, \thr)$. Therefore, the probability that $e = (v, w)$ is chosen in Step~\ref{step:h_w} satisfies
    \begin{equation}
      \label{eqn:heavy-return}
      \Pr[e \text{ is returned}] = \frac{\dm(v)}{n \, \thr} \cdot \frac{1}{d(v)} > \frac{1}{n \, \thr} \paren{1-\frac{m}{\thr^2}},
    \end{equation}
    where the inequality follows from Equation~\eqref{eqn:light-degree}. On the other hand, $\dm(v) \leq d(v)$ implies that
    \[
    \Pr[e \text{ is returned}] \leq \frac{1}{n \, \thr}\,.
    \]
    Finally, we bound the success probability by
    \[
    \Pr[\text{ Success }] = \Pr\left[\bigcup_{e \in \emh}\set{e \text{ is returned}}\right] = \sum_{e \in \emh} \Pr[e \text{ is returned}] > \frac{\abs{\emh}}{n \, \thr} \paren{1-\frac{m}{\thr^2}}.
    \]
    The second equality holds because the events $\set{e \text{ is returned}}$ are disjoint, while the inequality holds by Equation~\eqref{eqn:heavy-return}.
  \end{proof}

  The following lemma shows that we can sample an edge in $E$ almost uniformly by invoking $\sle$ or $\she$ with equal probability. The proof of the lemma demonstrates a tradeoff between the quality of the sample (i.e., closeness to uniformity) and the success probability. A larger threshold $\thr$ gives a closer to uniform sample, but is less likely to succeed. In particular, choosing $\thr = n-1$ gives a perfectly uniform sample, but the success probability is $O(m / n^2)$. 

  \begin{lem}
    \label{lem:sample-edge}\sloppy
    For any $\e$ satisfying $0 < \e < 1/2$ and for $\thr \geq \sqrt{2 m / \e}$, consider the procedure which with probability $1/2$ invokes $\sle(\thr)$, and with probability $1/2$ invokes $\she(\thr)$. The probability that an edge is successfully returned is at least $(1 - \e) m / (2 n \, \thr)$. If an edge is returned, then it is distributed according to a distribution $P$ that is pointwise  $\e$-close to uniform.
  \end{lem}
  \begin{proof}
    By Lemmas~\ref{lem:light} and~\ref{lem:heavy}, the probability of successfully returning an edge satisfies
    \begin{align*}
      \Pr[\text{ Success }] &= \frac 1 2 \Pr[\sle \text{ succeeds}] + \frac 1 2 \Pr[\she \text{ succeeds}]\\
      &> \frac 1 2 \cdot \frac{\abs{\eml}}{n \, \thr} + \frac 1 2 \cdot \frac{\abs{\emh}}{n \, \thr} \paren{1 - \frac{m}{\thr^2}}\\
      &\geq \frac{\abs{\eml}}{2 n \, \thr} + \frac{\abs{\emh}}{2 n \, \thr} (1 - \e / 2)\\
      &\geq (1 - \e) \frac{m}{2 n \, \thr}\;.
    \end{align*}
    The second inequality holds because $\thr \geq \sqrt{2 m / \e}$. Also by Lemmas~\ref{lem:light} and~\ref{lem:heavy}, the probability, $p_e$, that a specific edge $e$ is returned satisfies
    \[
    \frac{1 - \e/2}{2 n \, \thr} \leq p_e \leq \frac{1}{2 n \, \thr}.
    \]
    Thus, the distribution on sampled edges satisfies
    \[
    1 - \e / 2 \leq \frac{P(e)}{P(e')} \leq \frac{1}{1 - \e / 2} \leq 1 + \e\,, \quad \text{for all } e,e' \in E.
    \]
    Therefore, $P$ is pointwise  $\e$-close to uniform by Lemma~\ref{lem:epsilon-close}.
  \end{proof}
  
  \begin{figure}[htb!]
    \fbox{
      \begin{minipage}{0.9\textwidth}
	$\saue(n,\mhat, \e)$ \label{saue}
	\smallskip
	\begin{compactenum}
        \item Set $\thr = \sqrt{2 \mhat / \e}$.
	\item For $i=1$	 to $q=\setq$ do: \label{step:q}
	  \begin{compactenum}
	  \item With probability $1/2$ invoke $\sle(\thr)$ and with probability $1/2$ invoke $\she(\thr)$.
	  \item If an edge $(u,v)$ was returned, then \textbf{return} $(u,v)$.
	  \end{compactenum}
	\item \textbf{Return} \emph{fail}.
	\end{compactenum}
      \end{minipage}
    }
  \end{figure}

\sloppy
  By Lemma~\ref{lem:sample-edge}, we can sample an almost uniform edge from $E$:
  repeatedly call $\sle$ or $\she$ each with probability $1/2$ until an edge is returned. The following lemma shows that $\saue$ returns an edge almost uniformly so long as a good estimate of $m$ is known to the algorithm. We can obtain such an estimate by applying the algorithm of Goldreich \& Ron (Theorem~\ref{thm:gr08}) without increasing the asymptotic query cost of our algorithm.\footnote{If one wishes sample multiple edges, the approximate value of $m$ should be re-computed for each sample.} We note that if the number of invocations $q$ (set in Step~\ref{step:q}) is greater than $n$ then we can simply revert to an algorithm that samples a uniform vertex and query for its $i\th$ neighbor where $i$ is drawn uniformly in $[n]$.

  \begin{lem}
    \label{lem:saue}
    If $m \leq \mhat \leq 2m$, then Algorithm $\saue$ returns an edge $(u,v) \in E$ with probability at least $2/3$. The distribution induced by the algorithm is pointwise $\e$-close to uniform over $E$.
  \end{lem}
  
  \begin{proof}
    Let $\chi_i$ be the indicator variable for the event that an edge $(u,v)$ was returned in the $i\th$ step of the for loop of the algorithm. By Lemmas~\ref{lem:light} and~\ref{lem:heavy} if $m \leq \mhat\leq 2m$, then for every $i=1,\ldots,q$, it holds that 
    \[
    \Pr[\chi_i=1] > (1 - \e) \frac{m}{2 n \theta} \geq (1 - \e)\frac{\sqrt{m \e}}{4 n}\;.
    \]
    Hence, if $m \leq \mhat\leq 2m$, then the probability that no edge is returned in $q=\setq$ invocations is at most 
    \[
    \left(1-\frac{(1-\eps)\sqrt{\eps m}}{4n}\right)^{\setq} < 1/3.
    \]
    If $m \leq \mhat\leq 2m$, then it also holds that $\thr\geq \sqrt{2m/\e}$, and by Lemma~\ref{lem:sample-edge} the  returned edge is sampled from a distribution which is pointwise  $\e$-close to uniform.
  \end{proof}

  Finally, if $m$ is not known to the algorithm, we invoke Theorem~\ref{thm:gr08} to obtain an estimate $\mhat$ of $m$.

  \begin{proof}[Proof of Theorem~\ref{thm:alg}]
    Let $\mhat$ be an estimate of $m$. We call $\mhat$ a \emph{good} estimate if it satisfies $m \leq \mhat \leq 2m$. By repeating the algorithm of Goldreich \& Ron (Theorem~\ref{thm:gr08}) $O(\log(n / \e))$ times, then taking $\mhat$ to be the median value reported, a straightforward application of Chernoff bounds guarantees that $\mhat$ is good with probability at least $1 - \e / 2 n^2$.
    If $\mhat$ is good, by Lemma~\ref{lem:saue}, calling $\saue(n, \mhat,  \e/2)$ will successfully return an edge $e$ with probability at least $2/3$, and the returned edge is distributed according to a distribution $P$ which is pointwise  $\e/2$-close to uniform.

    Let $Q$ be the distribution of returned edges. If $\mhat$ is not good, we have no guarantee of the success probability of returning an edge, nor of the distribution from which the edge is drawn. However, since $\mhat$ is bad with probability at most $\e / 2 n^2$, for each $e$ we can bound
    \[
    \abs{Q(e) - U(e)} \leq \abs{Q(e) - P(e)} + \abs{P(e) - U(e)} \leq \frac{\e}{2 n^2} + \frac{\e}{2 m} \leq \frac{\e}{m} = \e U(e).
    \]
    Thus $Q$ is pointwise  $\e$-close to uniform.

    We now turn to analyze the expected query cost of the algorithm. By Theorem~\ref{thm:gr08}, the query cost of Goldreich \& Ron's algorithm is $\sO(n/\sqrt m)$. By Lemmas~\ref{lem:light} and~\ref{lem:heavy}, the procedures $\sle$ and $\she$ each performs a constant number of queries per invocation. Hence, the query cost of the algorithm is $O(q) = O\left({n}/{\sqrt{\eps \mhat}}\right)$. If $\mhat$ is good then $q$ is at most  $O\left({n}/{\sqrt{\eps m}}\right)$, and otherwise it is at most $O(n)$. Since $m$ is good with probability at least $1-\eps/2n^2$, it follows that the expected query cost is $\sO(n/\sqrt{\eps m})$.
    \end{proof}
  
  \section{A lower bound}
  \label{sec:lb}

  In this section we prove Theorem~\ref{thm:lb}, namely that any algorithm $\mA$ that samples an edge almost uniformly must perform $\Omega\left(n / \sqrt{m}\right)$ queries, even if $\mA$ is given $m$, the number of edges in the graph and is allowed to perform pair queries. 

  \begin{proof}[Proof of Theorem~\ref{thm:lb}]
    The result is trivial if $m = \Omega(n^2)$, so we assume that $m = o(n^2)$. Suppose there exists an algorithm $\mA$ that performs $t$ queries and with probability at least $2/3$ returns an edge sampled from a distribution $P$ satisfying $\tvd(P, U) \leq \e$ for $\e < 1/3$. Let $G'$ be an arbitrary graph, and let $n'$ and $m'$ denote the number of vertices and edges, respectively, in $G'$. Let $K$ be a clique on $k = \sqrt{2 m'}$ nodes. Let $G = G' \cup K$ be the disjoint union of $G'$ and $K$, and let $V_K$ and $E_K$ denote the vertices and edges of $K$ in $G$. Thus $G$ has $n = n' + k$ nodes, $m  > 2 m'$ edges, and $E_K$ contains at least $m / 2$ edges. The remainder of the proof formalizes the intuition that since $\mA$ makes relatively few queries, it is unlikely to sample vertices in $V_K$. Thus $\mA$ must sample edges from $E_K$ with probability significantly less than $1/2$.

    Assume that the vertices are assigned distinct labels from $[n]$ uniformly at random, independently of any decisions made by the algorithm $\mA$. Let $q_1, \ldots, q_t$ denote the set of queries that the algorithm performs, and let $a_1, \ldots, a_t$ denote the corresponding answers. We say that a query-answer pair $(q_i,a_i)$ is a \emph{witness pair} if (1) $q_i$ is a degree query of $v \in V_K$, or (2) $q_i$ is a neighbor query for some $v \in V_k$, or (3) $q_i$ is a pair query for some $(v, w) \in E_K$. For $i \in [t]$ let $\NW_i$ denote that event that $(q_1,a_1),\ldots,(q_{i},a_{i})$ are not witness pairs, and let $\NW = \NW_t$. Let $A_K$ be the event that $\mA$ returns some edge $e \in E_K$. Since $\tvd(P, U) < \e$, we must have $\abs{\Pr_P[A_K] - \Pr_U[A_K]} \leq \e$. Since $\abs{E_K} \geq m / 2$, we have $\Pr_U[A_K] \geq 1/2$, hence
    \begin{equation}
      \label{eqn:ak-lb}
      \Pr_P[A_k] \geq \frac 1 2 - \e.
    \end{equation}
    The law of total probability gives
    \begin{equation}
      \label{eqn:ak-tp}
      \Pr_P[A_K] = \Pr_P[\NW] \cdot \Pr_P[A_K \mid \NW] + \Pr_P[\NW^c] \cdot \Pr_P[A_K \mid \NW^c]\,,
    \end{equation}
    where $\NW^c$ denotes the complement of the event $\NW$.
    \begin{description}
    \item[Claim.] $\Pr_P[A_K \mid \NW] = o(1)$.
    \item[Proof of Claim.] Suppose the event $\NW$ occurs, i.e., $\mA$ does not observe a witness pair after $t$ queries. Thus, if $\mA$ returns an edge $e = (u, v) \in E_K$, it cannot have made queries involving $u$ or $v$. We can therefore bound
      \[
      \Pr_P[A_K \mid \NW] \leq \abs{E_K} \frac{1}{(n - 2 t)^2} \leq \frac{m}{2 (n - 2 t)^2} \leq \frac{2 m}{n^2}\,.
      \]
      The first inequality holds because the identities of any $u, v \in V_K$ are uniformly distributed among the (at least) $n - 2 t$ vertices not queried by $\mA$. The second inequality holds assuming $t < n / 4$. The claim follows from the assumption that $m = o(n^2)$.
    \end{description}
    Combining the result of the claim with equations~(\ref{eqn:ak-lb}) and~(\ref{eqn:ak-tp}) gives
    \begin{equation}
      \label{eqn:bound}
      \frac 1 2 - \e \leq \Pr_P[A_k] = \Pr_P[\NW] \cdot \Pr_P[A_K \mid \NW] + \Pr_P[\NW^c] \cdot \Pr_P[A_K \mid \NW^c] \leq \Pr_P[\NW^c] + o(1)\,.
    \end{equation}
    We bound $\Pr_P[\NW^c]$ by
    \begin{align*}
      \Pr[\NW^c] &= \Pr\left[\bigcup_{i \leq t} \set{(q_i, a_i) \text{ is the first witness pair}}\right]\\
      &= \sum_{i \leq t} \Pr[(q_i, a_i) \text{ is a witness pair} \mid \NW_{i-1}]\\
      &\leq \sum_{i \leq t} \frac{2 k}{n - 2 i} \leq \frac{4 k t}{n} \leq t \frac{4 \sqrt{2 m}}{n}.
    \end{align*}
    Combining this bound with (\ref{eqn:bound}) and solving for $t$ gives
    \[
    \frac 2 3 \paren{\frac 1 2 - \e - o(1)}\frac{n}{4 \sqrt{2 m}} < t\,.
    \]
    The factor of $2/3$ is because $\mA$ returns an edge with probability at least $2/3$. Thus, $t = \Omega(n / \sqrt{m})$, as desired.
  \end{proof}

  \section{Discussion}
  \label{sec:discussion}
%











  In this paper, we provided asymptotically matching upper and lower bounds for the problem of sampling an edge almost uniformly from a graph. While the approximation guarantee of our algorithm is strong, we do not know if it is possible to sample from an exactly uniform distribution in sublinear time.

  \begin{que}
  Is it possible to sample an edge exactly uniformly using $o(n)$ vertex, degree, neighbor, and pair queries?
  \end{que}

  As noted in the introduction, sampling edges (nearly) uniformly is equivalent to sampling each vertex with probability proportional to its degree. This distribution on vertices coincides with the stationary distribution for a random walk on a non-bipartite, connected, undirected graph (see, for example~\cite{Levin2009, Sinclair2012}). A random walk is similar to the algorithms we consider where the queries consist of a single vertex query (which may be taken from an arbitrary distribution), followed by many random neighbor queries. Thus, the mixing time of a random walk gives an upper bound on the query cost of sampling an edge almost uniformly (in total variational distance). 

  For many families of graphs, the mixing time may be significantly faster than the algorithm we proposed. However, the guarantee provided by our algorithm is independent of the graph's topology. When a graph is slowly mixing, our algorithm converges to the stationary distribution faster than a random walk precisely because our algorithm performs multiple vertex queries. We think it is interesting to study tradeoffs between the number of allowed vertex queries and neighbor queries. In general, we believe it would be valuable to consider a model in which different queries have predetermined costs, and the goal is to minimize the query cost of the algorithm for particular families of graphs.

  \bibliographystyle{plain}
  \bibliography{sampling-edges}

\end{document}